\documentclass[a4paper]{article}
\usepackage[utf8]{inputenc}
\usepackage{amsmath}
\usepackage{amssymb}
\usepackage{amsthm}
\usepackage{authblk} 
\usepackage{color}
\usepackage{graphicx}
\usepackage{physics}
\usepackage[colorinlistoftodos]{todonotes}
\usepackage[colorlinks=true, allcolors=blue]{hyperref}
\usepackage[numbers]{natbib}

\newcommand{\Span}{\operatorname{span}}
\newcommand{\Rank}{\operatorname{rank}}

\newtheorem{theorem}{Theorem}[section]
\newtheorem{lemma}[theorem]{Lemma}

\newtheorem{corollary}[theorem]{Corollary}
\newtheorem{definition}[theorem]{Definition}
\newtheorem{example}[theorem]{Example}
\newtheorem{remark}[theorem]{Remark}

\newcommand{\iddots}{\rotatebox{80}{$\ddots$}}

\usepackage[a4paper,top=3cm,bottom=2cm,left=3cm,right=3cm,marginparwidth=1.75cm]{geometry}

\title{A synchronous NPA hierarchy with applications}
\author{Travis B. Russell}
\date{}

\usepackage{graphicx}

\begin{document}

\maketitle

\begin{abstract}
    We present an adaptation of the NPA hierarchy to the setting of synchronous correlation matrices. Our adaptation improves upon the original NPA hierarchy by using smaller certificates and fewer constraints, although it can only be applied to certify synchronous correlations. We recover characterizations for the sets of synchronous quantum commuting and synchronous quantum correlations. For applications, we show that the existence of symmetric informationally complete positive operator-valued measures and maximal sets of mutually unbiased bases can be verified or invalidated with only two certificates of our adapted NPA hierarchy.
\end{abstract}

\section{Introduction}

Technological advances in quantum computing and quantum communication have accelerated in recent years, putting a number of high-stakes applications in the realm of the potential near future. One such application is quantum key distribution, a protocol in which a secret key is distributed to two distant parties through the measurement of entangled particles. The security of device-independent quantum key distribution is based on the laws of quantum mechanics when entanglement is present in the particles measured \cite{Umesh_Vidick_SecurityQKD}. Moreover, this entanglement can be verified by considering the probability distributions generated by the measurement devices used in the key generation process. These probability distributions are called \textbf{quantum correlations}. However, many open questions remain regarding precisely which probability distributions can be certified as quantum correlations (e.g. see \cite{Fu_Miller_Slofstra_2021_preprint}). 

The best known method for distinguishing quantum correlations from other kinds of probability distributions is the NPA hierarchy, developed in \cite{NPA2008}. Roughly, the NPA hierarchy is an infinite sequence of semidefinite programs which yield positive semidefinite matrices certifying that a given probability distribution may be a quantum correlation. If the given probability distribution $p$ yields a complete infinite sequence of certificates, then that distribution is certified as a \textbf{quantum commuting} correlation, meaning that $p$ was potentially generated by a valid quantum measurement scenario according to the Haag-Kastler axioms of relativistic quantum mechanics \cite{HaagKastler}, though the Hilbert space required may have infinite dimension. In practice, one cannot generate an infinite sequence of certificates directly. However, if there exists a certificate $\Gamma^{m+1}$ extending the previous certificate $\Gamma^m$ and having the same rank, then the hierarchy can be terminated and the correlation can be certified as a quantum correlation arising from a finite dimensional Hilbert space. The NPA Hierarchy can also be developed using the theory of universal C*-algebras (see Section 3 of \cite{PaulsenEtAlSynchronous}), and it was recently generalized to the setting of prepare-and-measure scenarios (see \cite{npjPrepareMeasureCorrelations}).

The distinction between quantum commuting correlations and quantum correlations would be of less practical importance if it were possible to approximate an arbitrary quantum commuting correlation with a quantum correlation. The question of whether or not this was possible remained open for many years and generated tremendous research interest, eventually becoming tied to a long-standing problem in mathematics known as Connes' embedding problem. These questions were finally settled recently in the paper \cite{mipStarEqualRe}, which showed that some quantum commuting correlations cannot be approximated by quantum correlations. Their methods required only synchronous quantum correlations, which are the subject of this paper.

In this paper, we present an adaptation of the NPA hierarchy for certifying synchronous quantum and quantum commuting correlations. While a synchronous correlation can be verified using the original NPA hierarchy as well, our adaptation has some advantages. The certificates produced by the hierarchy are smaller than those produced in the original NPA hierarchy. Moreover, there are fewer linear constraints imposed on the certificate, as one only needs to check that the certificates satisfy a kind of cyclic symmetry. See Remark \ref{rmk: advantage} below for more details. Our adapted hierarchy yields new characterizations for the sets of synchronous quantum and quantum commuting correlations. To further motivate these tools, we demonstrate how one can verify or invalidate two major open problems in quantum information theory, namely the existence of symmetric informationally-complete positive operator-valued measures (SIC-POVMs) and maximal sets of mutually unbiased bases (MUBs) in each dimension, using only two certificates of our adapted NPA hierarchy.

We conclude this introduction by mentioning some related work. The certificates in our adapted NPA hierarchy are examples of \textit{tracial Hankel matrices}, positive-semidefinite matrices indexed over words in a finite alphabet satisfying certain cyclical constraints. These have been used previously in the literature to certify the existence of tracial states on C*-algebras (see, for example, Section 4.5 of \cite{KlepPohvTracialMoment} or the paper \cite{KlepBurgdorfTracial}). The existence of flat extensions is also considered in \cite{KlepPohvTracialMoment}. Consequently, we do not expect the results of Sections 3 and 4 to be surprising to experts. However, we are not aware of any reformulation of the NPA hierarchy for certifying the existence of synchronous correlations in the literature and we feel the advantages outlined in Remark \ref{rmk: advantage}, and especially the applications in Section 5, provide sufficient motivation for sharing the details of this reformulation. We thank J. W. Helton and the anonymous referee for pointing out these references to the author. 

Our results on SICs and MUBs rely on Theorem \ref{thm: SIC algebra} and Theorem \ref{thm: MUB Algebra} below. These are algebraic characterizations of the C*-algebras generated by projections associated to a SIC-POVM or a maximal family of mutually unbiased bases. The definitions are inspired by the similarly defined MUB algebra of \cite[Theorem 21.3]{Navascues2012}. The MUB algebra is used in \cite{GriblingPolakMUBs} to give another characterization for MUBs in terms of semidefinite programs. There, the authors use different techniques, relying on wreath product symmetries of certain groups to produce moment matrices. For MUBs in the six-dimensional Hilbert space, these moment matrices would correspond to at least order five certificates in our adapted NPA Hierarchy (see Table 2 and Table 3 of \cite{GriblingPolakMUBs}) as opposed to order two certificates using our approach (in any dimension). We thank the anonymous referee for pointing out these references to the author.

\section{Preliminaries}

We begin with an overview of the notation and mathematical prerequisites for the paper. We let $\mathbb{N}, \mathbb{R}$, and $\mathbb{C}$ denote the sets of positive integers, real numbers, and complex numbers, respectively. Given $\lambda \in \mathbb{C}$, we let $\overline{\lambda}$ denote its complex conjugate. For each $n \in \mathbb{N}$, we let $M_n$ denote the set of $n \times n$ matrices with entries in $\mathbb{C}$. 

For each $N \in \mathbb{N}$, let $[N] = \{1,2,\dots,N\}$. Given a set $A$, we let $A^*$ denote the set of all words in $A$, including the empty word which we denote $0$. For each $w \in A^*$, we let $|w|$ denote the length of $w$, with the convention that $|0|=0$. For every $k \in \mathbb{N}$, let $A^k$ denote the set of all words of length at most $k$. For example, $[N]^k$ denotes all words of length at most $k$ in the symbols $\{1,2,\dots,N\}$.

We assume basic familiarity with the theory of Hilbert spaces over $\mathbb{C}$ and bounded linear operators on Hilbert spaces. Given a Hilbert space $H$, we sometimes use the notation $\langle h, k \rangle$ to denote the inner product of vectors $h,k \in H$, and we assume the inner product is linear in the second component and conjugate-linear in the first. We also employ bra-ket notation whenever convenient, for example letting $\ket{\phi}, \ket{\psi}$ denote vectors in a Hilbert space and $\braket{\psi}{\phi}$ denote their inner product. We use the notation $\vec{v}$ whenever regarding vectors as column matrices in the finite-dimensional Hilbert space $\mathbb{C}^n$. We let $B(H)$ denote the set of operator norm bounded operators on a Hilbert space $H$, and we let $T^{\dagger}$ denote the adjoint of an operator $T \in B(H)$. By a \textbf{C*-algebra}, we mean a norm-closed $\dagger$-closed subalgebra of $B(H)$. A \textbf{state} on a unital C*-algebra $\mathfrak{A}$ is a linear functional $\phi: \mathfrak{A} \to \mathbb{C}$ mapping the identity to 1 and mapping positive elements of $\mathfrak{A}$ to positive real numbers. A state $\phi: \mathfrak{A} \to \mathbb{C}$ is \textbf{tracial} if $\phi(ab)=\phi(ba)$ for all $a,b \in \mathfrak{A}$. A state  $\phi$ is \textbf{faithful} if $\phi(x^{\dagger}x) > 0$ whenever $x \neq 0$. An element $P \in \mathfrak{A}$ is called a \textbf{projection} if $P = P^{\dagger} = P^2$. A set of projections $\{P_1, P_2, \dots, P_N\} \subseteq \mathfrak{A}$ is called a \textbf{projection-valued measure} if each $P_i$ is a projection and if $\sum_{i=1}^N P_i = I$. We use freely well-known results about C*-algebras and Hilbert space operators throughout the paper, and we refer the reader to \cite{ConwayOperatorTheoryBook} for an in-depth introduction to these topics.

\subsection{Quantum correlations}

Let $n,k \in \mathbb{N}$. A tuple of real numbers $\{p(a,b|x,y)\}_{a,b \in [k], x,y \in [n]}$ is a \textbf{correlation} if it satisfies the relation \[ \sum_{a,b \in [k]} p(a,b|x,y) = 1 \] for all $x,y \in [n]$. A correlation $p(a,b|x,y)$ is called \textbf{nonsignalling} if the quantities \[ p_A(a|x) = \sum_{b \in [k]} p(a,b|x,y) \quad \text{ and } \quad p_B(b|y) = \sum_{a \in [k]} p(a,b|x,y) \] are well-defined, meaning that the sum expressing $p_A(a|x)$ is independent of the choice of $y \in [n]$ and the sum expressing $p_B(b|y)$ is independent of the choice of $x \in [n]$. Nonsignalling correlations model a scenario where two parties, traditionally named Alice and Bob, are provided questions $x$ and $y$, respectively, from a referee. Without communicating with each other, Alice produces an answer $a$ and Bob produces an answer $b$ with probability $p(a,b|x,y)$. The lack of communication between Alice and Bob can be verified after many trials by checking that the quantities $p_A(a|x)$ and $p_B(b|y)$ are well-defined; i.e. by checking that $p(a,b|x,y)$ is nonsignalling.

Nonsignalling correlations arise in quantum communication protocols, such as quantum key distribution \cite{Umesh_Vidick_SecurityQKD}. In these scenarios, Alice and Bob produce their answers by performing measurements on particles emitted from a common source. These particles may be entangled, yielding observable differences from correlations which arise in classical scenarios \cite{CHSH}. Mathematically, a correlation $\{p(a,b|x,y)\}$ is called a \textbf{quantum correlation} if there exists a finite dimensional Hilbert space $H$, projection-valued measures $\{E_{x,a}\}_{a=1}^k, \{F_{y,b}\}_{b=1}^k \subset B(H)$, and a unit vector $\ket{\phi} \in H \otimes H$ such that 
\begin{equation}\label{eqn: Q Correlation defn} p(a,b|x,y) = \bra{\phi} E_{x,a} \otimes F_{y,b} \ket{\phi}. \end{equation} 
In this formulation, Alice and Bob apply measurements corresponding to the projection-valued measures $\{E_{x,a}\}_{a=1}^k$ and $\{F_{y,b}\}_{b=1}^k$ to their respective copies of the Hilbert space $H$ upon receiving questions $x$ and $y$, respectively, from the referee. The laws of quantum mechanics dictate that they obtain answers $a$ and $b$, respectively, with probability $p(a,b|x,y)$ as described in Equation (\ref{eqn: Q Correlation defn}).

Quantum correlations can be equivalently defined in terms of finite-dimensional C*-algebras. A correlation $\{p(a,b|x,y)\}$ is a quantum correlation if and only if there exists a finite dimensional C*-algebra $\mathfrak{A}$, projection-valued measures $\{E_{x,a}\}_{a=1}^k, \{F_{y,b}\}_{b=1}^k \subseteq \mathfrak{A}$ for which each $E_{x,a}$ commutes with each $F_{y,b}$, and a state $\phi: \mathfrak{A} \to \mathbb{C}$ such that $p(a,b|x,y) = \phi(E_{x,a} F_{y,b})$. If we eliminate the restriction that the C*-algebra be finite dimensional, we obtain a \textbf{quantum commuting correlation}. It was an open question for many years whether or not an arbitrary quantum commuting correlation can be approximated by quantum correlations \cite{Tsirelson1987}. Indeed, this question was shown to be equivalent to the Connes' embedding problem \cite{ConnesConjecture} of operator algebras (See \cite{JMPPSW2011}, \cite{FritzKirchberg}, and \cite{OzawaConnes}). The recent results of \cite{mipStarEqualRe} imply that some quantum commuting correlations cannot be approximated by quantum correlations, thus solving the Connes' embedding problem. 

\subsection{The NPA Hierarchy} \label{subsec: NPA}

\textbf{The NPA hierarchy} is an infinite sequence of semidefinite programs developed by Navascues-Pironio-Acin in \cite{NPA2008}. Each semidefinite program in the NPA hierarchy takes as input a correlation $\{p(ab|xy)\}$ and returns a certificate in the form of a positive-semidefinite matrix, provided that the semidefinite program with input $\{p(ab|xy)\}$ is feasible. It is shown in \cite{NPA2008} that a correlation $\{p(ab|xy)\}$ is quantum commuting if and only if every semidefinite program in the NPA hierarchy returns a postive-semidefinite certificate. When $\{p(ab|xy)\}$ is a quantum correlation, only finite many levels of the NPA hierarchy are needed to certify that $\{p(ab|xy)\}$ is a quantum correlation. However, there is no known efficient method for distinguishing quantum correlations for quantum commuting correlations using the NPA hierarchy. We summarize here the basic elements of the NPA hierarchy and refer the reader to \cite{NPA2008} for more details.

Let $A$ and $B$ be finite sets, and let $C = A \cup B$ denote the disjoint union of $A$ and $B$. The set $A$ will represent projections belonging to Alice and the set $B$ will represent projections belonging to Bob. Given a word $w \in C^*$, let $w_A = a_1^{\alpha_1} a_2^{\alpha_2} \dots a_n^{\alpha_n} \in A^*$ and $w_B = b_1^{\beta_1} b_2^{\beta_2} \dots b_m^{\beta_m} \in B^*$, where $w_A$ is obtained from $w$ by concatenating the letters of $w$ which belong to $A$ in order from left to right with $a_k \neq a_{k+1}$ for each $k$, and similarly obtaining $w_B$ from the remaining letters of $w$. If $w_A = a_1^{\alpha_1} a_2^{\alpha_2} \dots a_n^{\alpha_n} \in A^*$ and $w_B = b_1^{\beta_1} b_2^{\beta_2} \dots b_m^{\beta_m} \in B^*$, we define $r(w) = a_1 a_2 \dots a_n b_1 b_2 \dots b_m$. We write $w \sim w'$ if $r(w) = r(w')$. Finally, for any word $w = c_1 c_2 \dots c_{n-1} c_n \in C^*$ we let $w^{\dagger} = c_n c_{n-1} \dots c_2 c_1$.

Now fix $n \in \mathbb{N}$ and assume $A = \cup_{x=1}^n A_x$ and $B = \cup_{y=1}^n B_y$, where the sets $A_1,\dots,A_n,B_1,\dots,B_n$ are mutually disjoint. For words $w,v \in C^*$, we write $w \perp v$ if $r(w^{\dagger}v) = a_1 a_2 \dots a_m b_1 b_2 \dots b_{m'}$ with $a_i, a_{i+1} \in A_x$ for some $x \in [n]$ and some index $i$; or $b_j, b_{j+1} \in B_y$ for some $y \in [n]$, and some index $j$. Let $\Gamma^k = (\Gamma_{w,v})$ be a matrix indexed by words $w,v \in C^k$. Then $\Gamma^k$ is a \textbf{certificate of order $k$} for a tuple $\{p(a,b)\}_{a \in A^1, b \in B^1}$ if $\Gamma^k$ is positive-semidefinite and satisfies
\begin{enumerate}
    \item $\Gamma_{0,0} = 1$ (unitality),
    \item $\Gamma_{w,v} = \Gamma_{w'v'}$ whenever $w^{\dagger}v \sim (w')^{\dagger}v'$,
    \item $\Gamma_{w,v} = 0$ if $w \perp v$ (orthogonality)
\end{enumerate}
and $\Gamma_{a,b} = p(a,b)$ for all $a \in A^1$ and $b \in B^1$.

\begin{theorem}[\cite{NPA2008}] \label{thm: NPA}
Suppose that a tuple $\{p(a,b)\}$ has an order $k$ certificate for every $k \in \mathbb{N}$. Then there exists a C*-algebra $\mathfrak{A}$, projections $\{E_a : a \in A\}$ and $\{F_b: b \in B\}$ in $\mathfrak{A}$ with $[E_a,F_b] = 0$ for all $a \in A$ and $b \in B$, and a state $\phi: \mathfrak{A} \to \mathbb{C}$ such that
\[ p(a,b) = \phi(E_a F_b) \]
for all $a \in A^1$ and $b \in B^1$, where $E_0 = F_0 = I$. Moreover
\[ \sum_{a \in A_x} E_a \leq I \quad \text{and} \quad \sum_{b \in B_y} F_b \leq I \]
for all $x,y \in [n]$.
\end{theorem}

Suppose $\{p(a,b)\}$ satisfies the conditions of Theorem \ref{thm: NPA}. Let $\{E_a\}_{a \in A}$ and $\{F_b\}_{b \in B}$ be the corresponding projections and $\phi$ the corresponding state. Then for each $x,y \in [n]$, 
\[ \{E_a\}_{a \in A_x} \cup \{I - \sum_{a \in A_x} E_a\} \quad \text{and} \quad \{F_b\}_{b \in B_y} \cup \{I - \sum_{b \in B_y} F_b \} \]
define projection valued measures. On the other hand, if we are given projection-valued measures $\{E_{x,a}\}_{a=1}^{k}$ and $\{F_{y,b}\}_{b=1}^k$ in $\mathfrak{A}$ with $[E_{x,a}, F_{y,b}] = 0$ and a state $\phi$ on $\mathfrak{A}$, we can produce certificates by setting
\[ \Gamma_{\alpha,\beta} = \phi((P_{\alpha_1} \dots P_{\alpha_m})^{\dagger} P_{\beta_1} \dots P_{\beta_{m'}}) \]
where $\alpha, \beta$ are strings in the letters $\{e(x,a), f(y,b) : x,y \in [n], a,b \in [k-1] \} \cup \{0\}$, and where $P_{e(x,a)} = E_{x,a}$, $P_{f(y,b)} = F_{y,b}$, and $P_0 = I$. The missing projections $E_{x,k}$ and $F_{y,k}$ can be recovered as 
\[ I - \sum_{a \in A_x} P_a \quad \text{and} \quad I - \sum_{b \in B_y} P_b \]
where $A_x = \{e(x,a) : a \in [k-1]\}$ and $B_y = \{f(y,b) : b \in [k-1]\}$.

\begin{remark}
\emph{A quantum correlation can be equivalently determined by a sequence of certificates $\Gamma^k$ where the entries satisfy \textit{completeness} conditions corresponding to the relations
\[ \sum_{a=1}^k E_{x,a} = I \quad \text{and} \quad \sum_{b=1}^k F_{y,b} = I. \]
However, this imposes more constraints on the set of certificates than is needed, since the correlation can be determined entirely from the values $\{p(ab|xy): a,b \in [k-1]; x,y \in [n]\}$ together with the values of the marginal densities $\{p_A(a|x), p_B(b|y)\}$.}
\end{remark}

Let $\Gamma^k$ be a certificate of order $k$. Then $\Gamma^k$ has a \textbf{rank loop} if the submatrix $(\Gamma_{aw,bv})$ indexed by words $w,v \in C^{k-1}$ beginning with letters $a \in A^1$ and $b \in B^1$ has the same rank as the full matrix $\Gamma^k$.

\begin{theorem}[\cite{NPA2008}] \label{Thm: NPA with rank loop}
Suppose that the tuple $\{p(a,b)\}$ has an order $k$ certificate with a rank loop. Then there exists a finite-dimensional C*-algebra $\mathfrak{A}$, projections $\{E_a : a \in A\}$ and $\{F_b: b \in B\}$ in $\mathfrak{A}$ with $[E_a,F_b] = 0$ for all $a \in A$ and $b \in B$, and a state $\phi: \mathfrak{A} \to \mathbb{C}$ such that
\[ p(a,b) = \phi(E_a F_b) \]
for all $a \in A^1$ and $b \in B^1$, where $E_0 = F_0 = I$. Moreover
\[ \sum_{a \in A_x} E_a \leq I \quad \text{and} \quad \sum_{b \in B_y} F_b \leq I \]
for all $x,y \in [n]$.
\end{theorem}

Theorem \ref{Thm: NPA with rank loop} identifies quantum correlations among the set of quantum commuting correlations using rank loops. One instance in which a rank loop arises is when there exists a $k$-order certificate $\Gamma^k$ which is a flat extension of an order $k-1$ certificate $\Gamma^{k-1}$ (i.e. $\Rank(\Gamma^k) = \Rank(\Gamma^{k-1}))$. In this case, $\Gamma^k$ is identified as the submatrix of $(\Gamma_{aw,bv})$ with $a=b=0$, forcing the rank of the larger submatrix to equal the rank of $\Gamma^k$. As remarked in Section 4 of \cite{NPA2008}, there are no known efficient methods for producing such flat extensions.

\subsection{Synchronous correlations}

A correlation $p(a,b|x,y)$ is called \textbf{synchronous} if $p(a,b|x,x) = 0$ whenever $a \neq b$. The following characterization of synchronous quantum and quantum commuting correlations comes from \cite{PaulsenEtAlSynchronous}.

\begin{theorem}[Corollary 5.6 of \cite{PaulsenEtAlSynchronous}] \label{thm: sync corr}
Let $p(a,b|x,y)$ be a synchronous correlation. Then $p(a,b|x,y)$ is a quantum commuting (resp. quantum) correlation if and only if there exists a (resp. finite-dimensional) C*-algebra $\mathfrak{A}$, projection valued measures $\{E_{x,a}\}_{a=1}^k \subseteq \mathfrak{A}$, and a tracial state $\tau: \mathfrak{A} \to \mathbb{C}$ satisfying \[ p(a,b|x,y) = \tau(E_{x,a} E_{y,b}). \]
\end{theorem}
\noindent For another characterization in terms of affine slices of projections of the completely positive semidefinite cone, see Corollary 5.5 of \cite{Sikora2017LinearCF}.

\begin{remark} \label{rmk: faithful}
\emph{In Theorem \ref{thm: sync corr}, we may assume without loss of generality that $\tau$ is faithful. This is because whenever $\mathfrak{A}$ is a C*-algebra and $\tau: \mathfrak{A} \to \mathbb{C}$ is a tracial state, we can define a new tracial state $\widehat{\tau}$ on the C*-algebra $\mathfrak{B} = \mathfrak{A} / \mathcal{J}$, where
\[ \mathcal{J} = \{x \in \mathfrak{A} : \tau(x^*x) = 0 \} \]
by setting $\widehat{\tau}(x + \mathcal{J}) = \tau(x)$. The subspace $\mathcal{J} \subseteq \mathfrak{A}$ is a self-adjoint ideal in $\mathfrak{A}$ so that $\mathfrak{B}$ is a C*-algebra. It is clear that $\widehat{\tau}$ is faithful on $\mathfrak{B}$. If $P, Q \in \mathfrak{A}$ are projections, then $\widehat{P} := P + \mathcal{J}, \widehat{Q} := Q + \mathcal{J} \in \mathfrak{B}$ are projections and $\tau(PQ) = \widehat{\tau}(\widehat{P} \widehat{Q})$.}
\end{remark}

We will make use of a family of matrices which are closely related to the set of synchronous correlations. The following definitions were introduced in \cite{RordamMusatNonClosure}.

\begin{definition}
Let $n \in \mathbb{N}$. Let $D_{qc}(n)$ be the set of tuples of real numbers $\{p(x,y)\}_{x,y \in [N]}$ for which there exists a C*-algebra $\mathfrak{A}$ and projections $P_1, P_2, \dots, P_n \in \mathfrak{A}$, and a faithful tracial state $\tau: \mathfrak{A} \to \mathbb{C}$ such that $p(x,y) = \tau(P_x P_y)$ for each $x,y \in [n]$. We say that $\{p(x,y)\} \in D_q(n)$ if the same conditions are met, but with the restriction that $\mathfrak{A}$ is finite-dimensional.
\end{definition}

It was shown in \cite{RordamMusatNonClosure} that the set $D_q(n)$ (resp. $D_{qc}(n)$) is affinely isomorphic to the set of a synchronous quantum correlations (resp. quantum commuting correlations) with $n$ questions and $k=2$ answers. For $k > 2$, it was shown in \cite{RussellTwoOutcome20} and \cite{HarrisThesis} that a particular affine slice of the set $D_q(nk)$ (resp. $D_{qc}(nk)$) is affinely isomorphic to the set of synchronous quantum correlations (resp. quantum commuting correlations) with parameters $n$ and $k$. Consequently, characterizing the structure of the set $D_q(N)$ (resp. $D_{qc}(N)$) with $N=nk$ is equivalent to characterizing the structure of the set of synchronous quantum (resp. quantum commuting) correlations. Therefore, we will focus our attention for the rest of the paper on the sets $D_q(N)$ and $D_{qc}(N)$.

\section{A synchronous NPA hierarchy}

In this section, we will characterize, for each $N \in \mathbb{N}$, the set of correlations $D_{qc}(N)$ in terms of positive semidefinite matrices indexed over the set $[N]^*$.

\begin{lemma} \label{lem: infinite gram}
Let $N \in \mathbb{N}$ and let $\Gamma$ be a matrix indexed by words $[N]^*$. Suppose that for each $k \in \mathbb{N}$, the finite matrix $\Gamma^k = (\Gamma_{\alpha, \beta})_{\alpha, \beta \in [N]^k}$ is positive-semidefinite. Then there exist a sequence of finite dimensional Hilbert spaces $H_1, H_2, \dots$ and a sequence of isometric linear maps $W_1: H_1 \to H_2, W_2: H_2 \to H_3, \dots$ such that for every $k \in \mathbb{N}$, 
\begin{enumerate}
    \item the Hilbert space $H_k$ is spanned by a set of vectors $\{ \ket{\alpha,k} : \alpha \in [N]^k \} \subseteq H_k$
    \item for every $\alpha \in [N]^k$, $W_k \ket{\alpha,k} = \ket{\alpha,k+1}$, and
    \item for every $\alpha, \beta \in [N]^k$, $\Gamma_{\alpha,\beta} = \braket{\beta,k}{\alpha,k}$.
\end{enumerate}
\end{lemma}

\begin{proof}
Let $k \in \mathbb{N}$. By the Gram decomposition of the postive semidefinite matrix $\Gamma^k$, there exists a finite dimensional Hilbert space $H_k$ and vectors $\{\ket{\alpha,k} : \alpha \in [N]^k\} \subseteq H_k$ spanning $H_k$ such that for every $\gamma, \beta \in [N]^k$, $\braket{\gamma,k}{\beta,k} = \Gamma_{\gamma, \beta}$. Likewise, there exists a Hilbert space $H_{k+1}$ spanned by vectors $\{ \ket{\beta,k+1} : \beta \in [N]^{k+1}\}$ such that for every $\gamma, \beta \in [N]^{k+1}$, $\braket{\gamma,k+1}{\beta,k+1} = \Gamma_{\gamma, \beta}$. Define a function $W_k$ from from the set $\{\ket{\alpha,k} : \alpha \in [N]^k\}$ to the set $\{\ket{\alpha, k+1}, \alpha \in [N]^k\}$ by $W_k \ket{\alpha,k} = \ket{\alpha,k+1}$ for each $\alpha \in [N]^k$. We first show that $W_k$ extends to a linear map from $H_k$ to $H_{k+1}$. To see this, observe that for every set of scalar coefficients $\{ t_{\alpha}, r_{\beta} : \alpha, \beta \in [N]^k\}$,
\begin{eqnarray}
\langle (\sum t_{\alpha} \ket{\alpha,k+1}), (\sum r_{\beta} \ket{\beta,k+1}) \rangle & = & \sum \overline{t_{\alpha}} r_{\beta} \braket{\alpha, k+1}{\beta, k+1} \nonumber \\
& = & \sum \overline{t_{\alpha}} r_{\beta} \braket{\alpha, k+1}{\beta, k+1} \nonumber \\
& = & \sum \overline{t_{\alpha}} r_{\beta} \Gamma_{\alpha,\beta} \nonumber \\
& = & \sum \overline{t_{\alpha}} r_{\beta} \braket{\alpha,k}{\beta,k} \nonumber \\
& = & \langle (\sum t_{\alpha} \ket{\alpha,k}), (\sum r_{\beta} \ket{\beta,k}) \rangle. \nonumber
\end{eqnarray}
Thus, if $\sum t_{\alpha} \ket{\alpha,k} = 0$, then \[ 0 = \langle (\sum t_{\alpha} \ket{\alpha,k}), (\sum t_{\alpha} \ket{\alpha,k}) \rangle = \langle (\sum t_{\alpha} \ket{\alpha,k+1}), (\sum t_{\alpha} \ket{\alpha,k+1}) \rangle. \]
It follows that setting $W_k(\sum t_{\alpha} \ket{\alpha,k}) = \sum t_{\alpha} \ket{\alpha,k+1}$ yields a well-defined linear extension of $W_k$. To see that $W_k$ is an isometry from $H_k$ to $H_{k+1}$, it suffices to check that $W_k^{\dagger} W_k$ is the identity on $H_k$. This follows from the observation that for every set of scalar coefficients $\{ t_{\alpha}, r_{\beta} : \alpha, \beta \in [N]^k\}$
\begin{eqnarray}
\langle (\sum t_{\alpha} \ket{\alpha,k}), W_k^{\dagger} W_k (\sum r_{\beta} \ket{\beta,k}) \rangle & = & \langle W_k (\sum t_{\alpha} \ket{\alpha,k}), W_k (\sum r_{\beta} \ket{\beta,k}) \rangle \nonumber \\
& = & \langle (\sum t_{\alpha} \ket{\alpha,k+1}), (\sum r_{\beta} \ket{\beta,k+1}) \rangle \nonumber \\
& = & \langle (\sum t_{\alpha} \ket{\alpha,k}), (\sum r_{\beta} \ket{\beta,k}) \rangle. \nonumber
\end{eqnarray} 
So $W_k^{\dagger} W_k$ is the identity map on $H_k$. \end{proof}

We briefly describe the construction for an inductive limit of a sequence of finite dimensional Hilbert spaces. Let $\{(H_k,W_k)\}_{k=0}^\infty$ be a sequence of pairs, each pair consisting of a finite dimensional Hilbert space $H_k$ and an isometry $W_k: H_k \to H_{k+1}$. Whenever $k < l$ we let $W_{k,l} := W_{l-1} W_{l-2} \dots W_{k+1} W_k: H_k \to H_l$. Let $\widehat{H}$ denote the disjoint union $\cup_k H_k$. Then we can define a pre-inner product on $\widehat{H}$ via $\langle x_l, x_k \rangle = \langle x_l, W_{k,l} x_k \rangle$ for each $x_l \in H_l$ and $x_k \in H_k$ when $k \leq l$ and $\langle x_l, x_k \rangle = \langle W_{l,k} x_l, x_k \rangle$ for each $x_l \in H_l$ and $x_k \in H_k$ when $l < k$. Let $\mathcal{N} = \{ x \in \widehat{H} : \langle x, x \rangle = 0\}$. Let $\lim_k H_k$ denote the completion of $\widehat{H}/\mathcal{N}$ with respect to this inner product. Then $\lim_k H_k$ is a Hilbert space with dimension $\lim_k \text{dim}(H_k)$. Moreover, for each $k \in \mathbb{N}$, there exists a natural isometry $V_k: H_k \to \lim_k H_k$ such that $V_{l} W_{k,l} H_k = V_k H_k$ for each $k < l$. Informally, we can use the $W_k$'s to identify $H_k$ as a subspace of $H_{k+1}$ and $V_k$ to identify $H_k$ as a subspace of $\lim_k H_k$, so that we have
$H_0 \subseteq H_1 \subseteq H_2 \subseteq \dots \subseteq \lim_k H_k$.

From the above construction and Lemma \ref{lem: infinite gram} we get the following corollary.

\begin{corollary} \label{cor: infinite gram}
Let $\Gamma$ be a matrix indexed by words in $[N]^*$. Assume that for each $k \in \mathbb{N}$, the finite matrix $\Gamma^k = (\Gamma_{\alpha, \beta})_{\alpha, \beta \in [N]^k}$ is positive-semidefinite. Then there exists a Hilbert space $H$ and vectors $\{ \ket{\alpha} : \alpha \in [N]^*\} \subseteq H$ such that for each $\alpha, \beta \in [N]^*$, $\Gamma_{\alpha,\beta} = \braket{\alpha}{\beta}$.
\end{corollary}

As in the original NPA hierarchy, we will be interested in positive semidefinite matrices $\Gamma$ indexed by words in $[N]^*$ whose entries satisfy certain relations. We will keep track of these relations by introducing an equivalence relation $\sim$ on $[N]^* \times [N]^*$. In the following, for each $\gamma \in [N]^*$ with $\gamma = g_1 g_2 \dots g_k$ and each permutation $\sigma$ of the set $[k]$, we let $\sigma(\gamma)$ denote the word $g_{\sigma(1)} g_{\sigma(2)} \dots g_{\sigma(k)}$. We define $\gamma^{\dagger} := g_k g_{k-1} \dots g_2 g_1$; i.e. $\gamma^{\dagger}$ is the word $\gamma$ written in reverse order.

\begin{definition} \label{defn: tracial relations}
Let $\alpha \in [N]^k$ and assume $\alpha = a_1^{r_1} a_2^{r_2} \dots a_n^{r_n}$ where $r_1, r_2, \dots, r_n \in \mathbb{N}$ with $\sum_i r_i \leq k$ and $a_i \neq a_{i+1}$ for each $i=1,2, \dots, n-1$. Then we define $\eta(\alpha) := a_1 a_2 \dots a_n \in [N]^n$ when $a_1 \neq a_n$ and $\eta(\alpha) := a_1 a_2 \dots a_{n-1} \in [N]^{n-1}$ otherwise. Given pairs $(\alpha, \beta), (\gamma, \delta) \in [N]^* \times [N]^*$, we say that $(\alpha, \beta) \sim (\gamma, \delta)$ if and only if $\eta(\alpha^{\dagger} \beta) = \sigma (\eta(\gamma^{\dagger} \delta))$ for some cyclic permutation $\sigma$.
\end{definition}

\begin{example} \label{ex: tracial relations}
\emph{We have $(32,1412) \sim (3221,14)$, since $\eta((32)^\dagger(1412)) = \eta(231412) = 23141$ and $\eta((3221)^\dagger(14) = \eta(122314) = 12314$, which are equivalent by a cyclic permutation.} 
\end{example}

\begin{remark} \label{rmk: tracial relations}
\emph{Definition \ref{defn: tracial relations} is motivated by the properties of projections and tracial states. For example, suppose that $\mathfrak{A}$ is a C*-algebra, $\tau: \mathfrak{A} \to \mathbb{C}$ is a tracial state, and $P_1, P_2, P_3, P_4 \in \mathfrak{A}$ are projections. Then
\[ \tau((P_3 P_2)^{\dagger} P_1 P_4 P_1 P_2) = \tau(P_2 P_3 P_1 P_4 P_1 ) = \tau(P_1 P_2 P_3 P_1 P_4) = \tau((P_3 P_2 P_2 P_1)^{\dagger} P_1 P_4). \]
This equality corresponds to the relation $(32,1412) \sim (3221,14)$ demonstrated in Example \ref{ex: tracial relations}.}
\end{remark}

The following Theorem characterizes the elements of $D_{qc}(N)$.

\begin{theorem} \label{thm: Synchronous NPA hierarchy}
Let $N \in \mathbb{N}$. Then $\{p(x,y)\} \in D_{qc}(N)$ if and only if there exists an infinite matrix $\Gamma$ indexed by the elements of $[N]^*$ with $\Gamma_{0, 0} = 1$ satisfying the following properties:
\begin{enumerate}
    \item For each $k \in \mathbb{N}$, the finite matrix $\Gamma^k = (\Gamma_{\alpha, \beta})_{\alpha, \beta \in [N]^k}$ is positive semidefinite.
    \item Whenever $(\alpha, \beta) \sim (\delta, \gamma)$ we have $\Gamma_{\alpha, \beta} = \Gamma_{\delta, \gamma}$.
    \item For each $x,y \in [N]$ we have $p(x,y) = \Gamma_{x,y}$.
\end{enumerate}
\end{theorem}

\begin{proof}
First assume that $\{p(x,y)\} \in D_{qc}(N)$. Then there exists a C*-algebra $\mathfrak{A}$, projections \[ P_1, P_2, \dots, P_N \in \mathfrak{A}, \] and a tracial state $\tau: \mathfrak{A} \to \mathbb{C}$ such that for every $x,y \in [N]$, $p(x,y) = \tau(P_x P_y)$. For each $\alpha = a_1 a_2 \dots a_k \in [N]^*$, let $P_{\alpha} := P_{a_1} P_{a_2} \dots P_{a_k}$, and let $P_{0} := I$. For each $\alpha, \beta \in [N]^*$ let $\Gamma_{\alpha, \beta} = \tau(P_{\alpha}^{\dagger} P_{\beta})$. Then $\Gamma_{0, 0} = \tau(I) = 1$. To prove (1), it suffices to check that the matrix of products $(P_{\alpha}^{\dagger} P_{\beta})_{\alpha, \beta}$ is positive in $M_n(\mathfrak{A})$, where $n = |[N]^k|$, since $\tau$ is completely positive (c.f. Proposition 3.8 of \cite{paulsen2002completely}). However, this follows from the observation that \[ (P_{\alpha}^{\dagger} P_{\beta})_{\alpha, \beta} = R^{\dagger}R \] where $R \in M_{1,n}(\mathfrak{A})$ is the row operator given by $R = [ P_{\alpha_1} P_{\alpha_2} \dots P_{\alpha_n}]$ and $\{\alpha_1, \dots, \alpha_n\}$ is an enumeration of $[N]^k$. To prove (2), we observe that whenever $(\alpha, \beta) \sim (\gamma, \delta)$ we have $\tau(P_{\alpha}^{\dagger} P_{\beta}) = \tau(P_{\gamma}^{\dagger} P_{\delta})$ since $\tau$ is cyclic and each $P_i$ satisfies $P_i^2 = P_i$ (see Remark \ref{rmk: tracial relations}). It is clear that (3) is satisfied. Therefore a matrix $\Gamma$ with the desired properties exists whenever $\{p(x,y)\} \in D_{qc}(N)$.

Now assume that we are given a matrix $\Gamma$ indexed over $[N]^*$ with $\Gamma_{0, 0} = 1$ and satisfying properties (1) and (2). For each $x,y \in [N]$, let $p(x,y) = \Gamma_{x,y}$. We will show that $\{p(x,y)\} \in D_{qc}(N)$. 

By Corollary \ref{cor: infinite gram}, there exists a Hilbert space $H$ and vectors $\{ \ket{\alpha} : \alpha \in [N]^*\} \subseteq H$ with dense span in $H$ such that for every $\alpha, \beta \in [N]^*$, $\Gamma_{\alpha, \beta} = \braket{\alpha}{\beta}$. For each $x \in [N]$, let $P_x$ denote the orthogonal projection onto the subspace of $H$ densely spanned by the vectors $\{ \ket{x \alpha}: \alpha \in [N]^*\}$. Clearly $P_x \ket{x \alpha} = \ket{x \alpha}$ for each $\alpha \in [N]^*$. Moreover, if $\alpha, \beta \in [N]^*$ then 
\begin{eqnarray}
\braket{x \beta}{\alpha} & = & \Gamma_{x \beta, \alpha} \nonumber \\
& = & \Gamma_{x \beta, x \alpha} \nonumber \\
& = & \braket{x \beta}{x \alpha} \nonumber
\end{eqnarray}
since $(x \beta, \alpha) \sim (x \beta, x \alpha)$. Since the range of $P_x$ is densely spanned by the set of vectors $\{\ket{x \beta} : \beta \in [N]^*\}$, we conclude that $P_x \ket{\alpha} = P_x \ket{x \alpha} = \ket{x \alpha}$ for each $\alpha \in [N]^*$.

As before, whenever $\alpha = a_1 a_2 \dots a_k \in [N]^k$, let $P_{\alpha}$ denote the product $P_{a_1} P_{a_2} \dots P_{a_k}$. Because $P_{x} \ket{\beta} = \ket{x \beta}$ for each $x \in [N]$ and $\beta \in [N]^*$, we see that $P_{\alpha} \ket{0} = \ket{\alpha}$ for each $\alpha \in [N]^*$. Hence $\Gamma_{\alpha, \beta} = \bra{0} P_{\alpha}^{\dagger} P_{\beta} \ket{0}$ for each $\alpha, \beta \in [N]^*$. Let $\mathfrak{A}$ denote the C*-algebra generated by the projections $P_1, \dots, P_N$ in $B(H)$ and define $\tau: \mathfrak{A} \to \mathbb{C}$ by $\tau(T) = \bra{0} T \ket{0}$ for each $T \in \mathfrak{A}$. Since $\braket{0}{0} = \Gamma_{0, 0} = 1$, $\tau$ defines a state on $\mathfrak{A}$. Furthermore, notice that for each $\alpha \in [N]^*$ and each cyclic permutation $\sigma$
\begin{eqnarray}
\tau(P_{\alpha}) & = & \bra{0} P_{\alpha} \ket{0} \nonumber \\
& = & \braket{0}{\alpha} \nonumber \\
& = & \Gamma_{0, \alpha} \nonumber \\
& = & \Gamma_{0, \sigma(\alpha)} \nonumber \\
& = & \braket{0}{\sigma(\alpha)} \nonumber \\
& = & \bra{0} P_{\sigma(\alpha)} \ket{0} \nonumber \\
& = & \tau(P_{\sigma(\alpha)}) \nonumber
\end{eqnarray} where we have used $(0, \alpha) \sim (0, \sigma(\alpha))$. It follows that $\tau$ is tracial on the $*$-algebra generated by the $P_x$'s and hence $\tau$ is a tracial state on $\mathfrak{A}$. If $\tau$ is not faithful, we can replace $\tau$ with a faithful tracial state on a quotient $\mathfrak{A} / \mathcal{J}$ of $\mathfrak{A}$ and replace each $P_x$ with $P_x + \mathcal{J}$, as described in Remark \ref{rmk: faithful}. Therefore the identification $p(x,y) := \Gamma_{x,y}$ defines a correlation $\{p(x,y)\} \in D_{qc}(N)$ since, for each $x,y \in [N]$, $p(x,y) = \tau(P_x P_y)$.
\end{proof}

Assume that $\{p(x,y)\} \in D_{qc}(N)$ and let $\Gamma$ be a positive semidefinite matrix as described in Theorem \ref{thm: Synchronous NPA hierarchy}. Then the submatrices $\Gamma^k = (\Gamma_{\alpha, \beta})_{\alpha, \beta \in [N]^k}$ each satisfy $\Gamma_{0,0}^k = 1$ and conditions 1, 2, and 3 of Theorem \ref{thm: Synchronous NPA hierarchy}. In general, any matrix $\Gamma^k$ indexed by the elements of $[N]^k$ is called a \textbf{certificate of order $k$} for $\{p(x,y)\}$ if $\Gamma_{0,0}^k = 1$ and:
\begin{enumerate}
    \item $\Gamma^k$ is positive semidefinite
    \item if $\alpha, \beta, \delta, \gamma \in [N]^k$ and $(\alpha,\beta) \sim (\delta, \gamma)$, then $\Gamma_{\alpha,\beta}^k = \Gamma_{\delta,\gamma}^k$, and
    \item for each $x,y \in [N]$, $\Gamma^k_{x,y} = p(x,y)$.
\end{enumerate}

\begin{corollary} \label{cor: Syn NPA Hierarchy}
Let $N \in \mathbb{N}$. Then $\{p(x,y)\} \in D_{qc}(N)$ if and only if there exists a sequence of certificates $\Gamma^1, \Gamma^2, \dots$ for $\{p(x,y)\}$.
\end{corollary}

\begin{proof}
Let $\{p(x,y)\}_{x,y \in [N]}$ be a tuple of real numbers. Suppose there exists a sequence of certificates $\Gamma^1, \Gamma^2, \dots$ for $\{p(x,y)\}$. We will establish the existence of a single infinite matrix $\Gamma$ indexed by words in $[N]^*$ which satisfies the conditions of Theorem \ref{thm: Synchronous NPA hierarchy}. To establish this, we mimick the arguments of Theorem 8 and Appendix B of \cite{NPA2008}, summarized here for the sake of completeness. Let $\widehat{\Gamma}^k$ denote the infinite matrix indexed by $[N]^*$ with
\[ \widehat{\Gamma}^k_{\alpha,\beta} = \begin{cases} \Gamma^k_{\alpha,\beta} & \alpha,\beta \in [N]^k \\ 0 & \text{else} \end{cases} \]
regarded as an element of $l^{\infty}([N]^*)$. We claim that the sequence $\{ \widehat{\Gamma}^k\}_{k=1}^\infty$ admits a convergent subsequence. This follows from the Banach-Alaoglu Theorem provided that the sequence $\{ \widehat{\Gamma}^k\}_{k=1}^\infty$ resides in the unit ball of $l^{\infty}([N]^*)$. To prove this, it suffices to establish that the diagonal elements of each certificate $\Gamma^k$ are bounded by 1, since each $\Gamma^k$ is positive semidefinite. For $m \in \mathbb{N}$ with $m < k$, let $\alpha \in [N]^m$ and $x \in [N]$. Then since the submatrix
\[ \begin{bmatrix} \Gamma_{\alpha, \alpha}^k & \Gamma_{\alpha, x \alpha}^k \\ \Gamma_{x \alpha, \alpha}^k & \Gamma_{x \alpha, x \alpha}^k \end{bmatrix} \]
is positive semidefinite, we have $|\Gamma_{x,\alpha}^k|^2 \leq \Gamma_{\alpha, \alpha}^k \Gamma_{x \alpha, x \alpha}^k$. Since $(x \alpha, x \alpha) \sim (x \alpha, \alpha)$, we have $\Gamma_{x,\alpha}^k = \Gamma_{x\alpha,x\alpha}^k$ and hence $\Gamma_{x \alpha, x \alpha}^k \leq \Gamma_{\alpha, \alpha}^k$. The claim follows by induction on $m$, since $\Gamma_{0,0}^k = 1$.

Conversely, if $\{p(x,y)\} \in D_{qc}(N)$, then Theorem \ref{thm: Synchronous NPA hierarchy} implies the existence of an infinite matrix $\Gamma$ indexed by $[N]^*$ for which the finite submatrices $\Gamma^k = (\Gamma_{\alpha, \beta})_{\alpha, \beta \in [N]^k}$ produce a sequence of certificates for $\{p(x,y)\}$.
\end{proof}

\begin{remark} \label{rmk: advantage}
\emph{We conclude this section by noting some potential advantages for using Theorem \ref{thm: Synchronous NPA hierarchy} to certify elements of $D_{qc}(N)$ rather than Theorem \ref{thm: NPA}. First, notice that an order $k$ certificate for the synchronous hierarchy is indexed over $[N]^k$, whereas the order $k$ certificate of the NPA hierarchy is indexed over words in $C^k$ where $C$ denotes the disjoint union of two copies of $[N]$. The set $[N]^k$ contains $\sum_{n=0}^k N^n$ elements, whereas $C^k$ contains $\sum_{n=0}^k (2N)^n$ elements. Thus the matrices considered in the synchronous hierarchy are smaller. This difference can be narrowed by recognising that many elements of $C$ are equivalent due to the commutativity property, but the certificates still remain larger in general. For example, $\Gamma^1$ in the synchronous hierarchy is an $(N+1) \times (N+1)$ matrix, while it is a $(2N+1) \times (2N+1)$ matrix in the original NPA hierarchy. Secondly, the orthogonality constraint that is needed in the original NPA hierarchy is redundant in the synchronous hierarchy. This is because whenever $\tau$ is a faithful tracial state and $P$ and $Q$ are projections, $\tau(PQ) = 0$ implies that $PQ = 0$, since $\tau(PQ)=\tau(QP^2Q) = \tau((PQ)^{\dagger}(PQ))$. On the other hand, if $\phi$ is a faithful (not necessarily tracial) state, then $\phi(PQ)=0$ does not imply that $PQ = 0$. Therefore the orthogonality condition must be imposed in the original hierarchy to ensure that projection-valued measures consist of mutually orthogonal projections. Finally, we note that the synchronous hierarchy also has the minor advantage that the first certificate $\Gamma^1$ is uniquely determined by the correlation $\{p(x,y)\}$. Given the matrix $\{p(x,y)\}$, we form the corresponding certificate $\Gamma^1$ by appending a single row and column corresponding to the empty word $0$. The entries $\Gamma_{0,x}$ and $\Gamma_{x,0}$ for $x \in [N]$ are uniquely determined since $(x,0) \sim (x,x) \sim (0,x)$ implies $\Gamma_{x,0}=\Gamma_{x,x}=\Gamma_{0,x}$. The entry $\Gamma_{0,0}$ is determined by the requirement $\Gamma_{0,0}=1$. In the original NPA hierarchy, there are $2N^2$ entries of the first certificate which are not determined by the correlation.}
\end{remark}

\section{The rank loop}

In the original NPA hierarchy, quantum correlations are distinguished from quantum commuting correlations by the existence of a rank loop in an order $k$ certificate $\Gamma^k$. As described in Subsection \ref{subsec: NPA}, a rank loop is a submatrix of the certificate $\Gamma^k$ with the same rank as $\Gamma^k$ and with indices of the form $(x \alpha, y \beta)$ where $\alpha$ and $\beta$ are words of length at most $k-1$, $x$ corresponds to one of Alice's projections and $y$ corresponds to one of Bob's projections. A rank loop also arises whenever $\Rank(\Gamma^{k-1}) = \Rank(\Gamma^k)$.

In the synchronous hierarchy, there is no need to index Alice and Bob's projections differently since they share the same set of projections. Thus the definition of the rank loop does not extend to the synchronous hierarchy directly. Instead, we say that an order $k$ certificate $\Gamma^k$ has a \textbf{rank loop} if the submatrix $\Gamma^{k-1} = (\Gamma^k_{\alpha, \beta})_{\alpha, \beta \in [N]^{k-1}}$ has the same rank as $\Gamma^k$ (i.e. $\Gamma^k$ is a flat extension of $\Gamma^{k-1}$). We now show that in the synchronous hierarchy, elements of $D_q(N)$ are characterized as correlations admitting an order $k$ certificate with a rank loop.

\begin{theorem} \label{thm: rank loop}
Let $N \in \mathbb{N}$. Then $\{p(x,y)\} \in D_q(N)$ if and only if there exists an integer $m \in \mathbb{N}$ and an order $m+1$ certificate $\Gamma^{m+1}$ with a rank loop. In particular, if there exists an order $m+1$ certificate for $\{p(x,y)\}$ with a rank loop, then there exists a C*-algebra $\mathfrak{A}$, projections $P_1, \dots, P_N \in \mathfrak{A}$, and a faithful tracial state $\tau: \mathfrak{A} \to \mathbb{C}$ such that, for every $x,y \in [N]$, $p(x,y) = \tau(P_x P_y)$ and such that $\mathfrak{A}$ is spanned by operators of the form $\{P_{a_1} P_{a_2} \dots P_{a_{m}} : a_1 a_2 \dots a_{m} \in [N]^m \}$ (where $P_0 := I$).
\end{theorem}

\begin{proof}
First assume that $\{p(x,y)\} \in D_{q}(N)$. Then there exists a finite dimensional C*-algebra $\mathfrak{A}$, projections $P_1, P_2, \dots, P_N \in \mathfrak{A}$ and a faitful tracial state $\tau: \mathfrak{A} \to \mathbb{C}$ such that for every $x,y \in [N]$, $p(x,y) = \tau(P_x P_y)$. We may assume without loss of generality that $\mathfrak{A}$ is generated by the projections $P_1, P_2, \dots, P_N$ as a C*-algebra. For each $\alpha \in [N]^*$ with $\alpha = a_1 a_2 \dots a_k$, set $P_{\alpha} = P_{a_1} P_{a_2} \dots P_{a_k}$, and let $\Gamma_{\alpha, \beta} = \tau(P_{\alpha}^{\dagger} P_{\beta})$ for each $\alpha, \beta \in [N]^k$. By the GNS construction for C*-algebras (c.f. Chapter 1, Section 7 of \cite{ConwayOperatorTheoryBook}), there exists a Hilbert space $H$, a unit vector $\ket{\phi} \in H$ and a $*$-homomorphism $\pi: \mathfrak{A} \to B(H)$ such that $\tau(P_{\alpha}^{\dagger} P_{\beta}) = \bra{\phi} \pi(P_{\alpha}^{\dagger} P_{\beta}) \ket{\phi}$ for each $\alpha, \beta \in [N]^*$. Since $\dim(\mathfrak{A}) < \infty$, there exists $m$ such that $\mathfrak{A}$ is spanned by $\{P_{\alpha} : \alpha \in [N]^m\}$. Let $\alpha_1, \alpha_2, \dots, \alpha_M$ be an enumeration of $[N]^m$, and let $\alpha_{M+1}, \alpha_{M+2}, \dots, \alpha_{M'}$ be an enumeration of $[N]^{m+1} \setminus [N]^{m}$. Then since
\[ \dim( \Span \{\pi(P_{\alpha}) \ket{\phi} : \alpha \in [N]^m \}) = \dim( \Span \{\pi(P_{\alpha}) \ket{\phi} : \alpha \in [N]^{m+1} \}) \]
we must conclude that $\Rank(\Gamma^m) = \Rank(\Gamma^{m+1})$, since
\[ \Gamma^m = \begin{bmatrix} \bra{\phi} P_{\alpha_1}^{\dagger} \\ \vdots \\ \bra{\phi} P_{\alpha_M}^{\dagger} \end{bmatrix} \begin{bmatrix} P_{\alpha_1} \ket{\phi} & \dots & P_{\alpha_M} \ket{\phi} \end{bmatrix} \quad \text{and} \quad \Gamma^{m+1} = \begin{bmatrix} \bra{\phi} P_{\alpha_1}^{\dagger} \\ \vdots \\ \bra{\phi} P_{\alpha_{M'}}^{\dagger} \end{bmatrix} \begin{bmatrix} P_{\alpha_1} \ket{\phi} & \dots & P_{\alpha_{M'}} \ket{\phi} \end{bmatrix}. \]

On the other hand, assume that $\Gamma^{m+1}$ is an order $m+1$ certificate for $\{p(x,y)\}$ with a rank loop. By Lemma \ref{lem: infinite gram}, there exists a Hilbert space $H$ and vectors $\ket{\alpha} \in H$ for each $\alpha \in [N]^{m+1}$ such that $\Gamma_{\alpha, \beta} = \braket{\alpha}{\beta}$ for each $\alpha, \beta \in [N]^{m+1}$. Since $\Rank(\Gamma^m) = \Rank(\Gamma^{m+1})$ we see that
\begin{equation} \label{eqn: rank loop span}
\dim( \Span \{ \ket{\alpha} : \alpha \in [N]^m\}) = \dim( \Span \{ \ket{\alpha} : \alpha \in [N]^{m+1}\}). 
\end{equation}
Therefore every vector $\ket{\alpha} \in H_{m+1}$ can be written as a linear combination of vectors of the form $\ket{ \alpha }$ where $\alpha \in [N]^m$. Hence, we may identify the Hilbert spaces $H_{m+1}$ and $H_m$. For each $x \in [N]$, let $P_x: H_m \to H_m$ denote the projection onto the subspace spanned by the vectors $\ket{x \alpha}$ for $\alpha \in [N]^m$. As shown in the proof of Theorem \ref{thm: Synchronous NPA hierarchy}, we have $P_x \ket{\alpha} = \ket{x \alpha}$ for each $\alpha \in [N]^m$. Let $\mathfrak{A}$ denote the finite-dimensional C*-algebra generated by the operators $P_x$ in $B(H_m)$. The proof that $\tau(T) = \bra{0} T \ket{0}$ for $T \in \mathfrak{A}$ defines a faithful trace on $\mathfrak{A}$ is identical to the argument presented in the proof of Theorem \ref{thm: Synchronous NPA hierarchy}. From Equation \ref{eqn: rank loop span}, it follows that for each $\alpha \in [N]^{m+1}$, $P_{\alpha} \in \Span \{P_{\beta} : \beta \in [N]^m\}$. This is because if $\alpha \in [N]^{m+1}$ and 
\[ \ket{\alpha} = \sum_{\beta \in [N]^m} t_{\beta} \ket{\beta} \]
then 
\[ \tau(( P_{\alpha} - \sum t_{\beta} P_{\beta})^{\dagger} ( P_{\alpha} - \sum t_{\beta} P_{\beta})) = \bra{0} ( P_{\alpha} - \sum t_{\beta} P_{\beta})^{\dagger} ( P_{\alpha} - \sum t_{\beta} P_{\beta}) \ket{0} = 0 \]
since $P_{\alpha} \ket{0} = \ket{\alpha} = \sum t_{\beta} P_{\beta} \ket{0}$. Thus $\mathfrak{A} = \Span \{P_{\beta} : \beta \in [N]^m\}$. We conclude that $\{p(x,y)\} \in D_q(N)$. \end{proof}

\section{Applications}

In this section, we consider two applications, each involving $d^2$ projections which span the vector space $M_d$. We begin by outlining how to characterize families of projections of this form. Throughout this section, recall that a \textbf{factor} is a C*-algebra $\mathfrak{A}$ with trivial center $Z(\mathfrak{A})$, meaning that if $T \in \mathfrak{A}$ commutes with every other element of $\mathfrak{A}$, then $T = \lambda I$ for some scalar $\lambda$. If $\mathfrak{A}$ is a finite-dimensional factor, then $\mathfrak{A} \cong M_d$ for some $d \in \mathbb{N}$.

Suppose that $P_1, \dots, P_N$ are projections which span $M_d$ where $N \geq d^2$. Let $\tau = \frac{1}{d}\Tr(\cdot)$ denote the unique tracial state on $M_d$ and let $P_0$ denote the identity. Then the matrix $\Gamma^2$ with entries
\[ \Gamma_{ab,xy} = \tau((P_a P_b)^{\dagger} P_x P_y) \]
indexed by $a,b,x,y \in [N]^2$ must have rank at most $d^2$. If the submatrix $(\Gamma_{a,b})_{a,b \in [N]^1}$ has rank $d^2$, then $\Gamma^2$ will have a rank loop and hence satisfy the conditions of Theorem \ref{thm: rank loop}. Also, because $M_d$ has trivial center, we know that if $T \in M_d$ and $[T,P_a] = 0$ for all $a \in [N]$, it follows that $T$ is a scalar multiple of the identity. This property may potentially be reflected by linear relations on the entries of $\Gamma^2$ (we will demonstrate this for the two cases we examine below).

Conversely, suppose we are given $\Gamma^2$ with $\Rank(\Gamma^2) = \Rank(\Gamma^1) = d^2$. By Theorem \ref{thm: rank loop}, there exists a $d^2$-dimensional C*-algebra $\mathfrak{A}$, projections $P_1, \dots, P_N \in \mathfrak{A}$ spanning $\mathfrak{A}$, and a faithful tracial state $\tau: \mathfrak{A} \to \mathbb{C}$ such that
\[ \Gamma_{ab,xy} = \tau((P_a P_b)^{\dagger} P_x P_y) \]
for every $a,b,x,y \in [N]$. If $\Gamma^2$ satisfies sufficiently many linear constraints to guarantee that the center of $\mathfrak{A}$ is trivial, then $\mathfrak{A}$ is a $d^2$-dimensional factor and hence $\mathfrak{A} \cong M_d$. Since $M_d$ has a unique faithful tracial state, $\tau = \frac{1}{d}\Tr(\cdot)$.

In the following, we consider two situations in which the matrix algebra $M_d$ may arise as a linear span of rank-one projections. In each situation, we will derive necessary and sufficient conditions on an associated certificate $\Gamma^2$ that guarantee the projections producing the certificate generate a $d^2$-dimensional C*-algebra $\mathfrak{A}$ with trivial center, implying that $\mathfrak{A} \cong M_d$. 

\subsection{SIC-POVMs} \label{subsec: SIC-POVM}

Let $d \in \mathbb{N}$. Then a set $\{P_1, P_2, \dots, P_{d^2}\}$ of rank one projections in $M_d$ is called a SIC-POVM if $\Span \{P_1, P_2, \dots, P_{d^2}\} = M_d$, $\sum_{i=1}^{d^2} P_i = d I_d$, and $\Tr(P_i P_j)=c$ for all $i \neq j$, where $c$ is a fixed positive constant. Under these conditions, it can be shown that \[ \Tr(P_i P_j) = \begin{cases} \frac{1}{d+1} & i \neq j \\ 1 & i = j \end{cases}. \]
It has been verified that SIC-POVMs exist in most dimensions $d \leq 50$, and numerical evidence suggests that they also exist in most dimensions $d \leq 150$. It is currently an open question whether or not SIC-POVMs exist in every dimension $d$, or if there is an upper bound on the dimension $d$ in which SIC-POVMs exist. See \cite{SICQuestion} for an overview of the history and open problems related to SIC-POVMs.

Define \[ p_{sic}^d(x,y) = \begin{cases} \frac{1}{d(d+1)} & x \neq y \\ \frac{1}{d} & x = y \end{cases}. \] We first verify that $p_{sic}^d(x,y)$ extends to a positive semidefinite certificate $\Gamma^1$ satisfying $\Rank(\Gamma^1) = d^2$. The certificate $\Gamma^1$ is uniquely defined and equals 
\begin{equation}
    \Gamma^1 = \begin{bmatrix} 1 & \frac{1}{d} & \frac{1}{d} & \dots & \frac{1}{d} \\
    \frac{1}{d} & \frac{1}{d} & \frac{1}{d(d+1)} & \dots & \frac{1}{d(d+1)} \\
    \frac{1}{d} & \frac{1}{d(d+1)} & \frac{1}{d} & & \frac{1}{d(d+1)} \\
    \vdots & \vdots & & \ddots & \vdots \\
    \frac{1}{d} & \frac{1}{d(d+1)} & \dots & & \frac{1}{d} \end{bmatrix}. \nonumber
\end{equation}
This matrix can be factored as
\[ \Gamma^1 = v v^T + \frac{1}{d+1} \begin{bmatrix} 0 & \vec{0}^T \\ \vec{0} & I \end{bmatrix}  - \frac{1}{d^2(d+1)} \begin{bmatrix} 0 & \vec{0}^T \\ \vec{0} & J \end{bmatrix} \]
where $v = \begin{bmatrix} 1 & \frac{1}{d} & \dots & \frac{1}{d} \end{bmatrix}^T \in M_{d^2+1,1}$, $\vec{0}$ denotes the zero matrix in $M_{d^2,1}$, $I$ denotes the $d^2 \times d^2$ identity matrix and $J$ denotes the $d^2 \times d^2$ matrix for which every entry is 1. Since the spectrum of $J$ is $\{0,d^2\}$, the spectrum of $\frac{1}{d+1} I - \frac{1}{d^2(d+1)} J$ is $\{\frac{1}{d+1}, 0\}$. It follows that $\Gamma^1$ is positive semidefinite. To see that $\Rank(\Gamma^1) = d^2$, notice that
\[ \frac{1}{d+1} I - \frac{1}{d^2(d+1)} J = \frac{1}{d+1} (I - \frac{1}{d^2} J) \]
and that $I - \frac{1}{d^2}J$ is a rank $d^2 - 1$ projection. Moreover,
\[ \frac{1}{d+1} \left( \begin{bmatrix} 0 & \vec{0}^T \\ \vec{0} & I \end{bmatrix}  - \frac{1}{d^2} \begin{bmatrix} 0 & \vec{0}^T \\ \vec{0} & J \end{bmatrix} \right) v = 0. \]
It follows that the rank of $\Gamma^1$ is $d^2$, since $v$, together with the $d^2-1$ eigenvectors for $\Gamma^1 - vv^T$, constitute a mutually orthogonal family of eigenvectors for $\Gamma^1$.

We now wish to consider certificates $\Gamma^2$ extending $\Gamma^1$ with rank $d^2$. We would like such a certificate to satisfy linear relations that guarantee the underlying C*-algebra generating $\Gamma^2$ is the matrix algebra $M_d$. The following theorem will allow us to find such relations.

\begin{theorem} \label{thm: SIC algebra}
Suppose $\mathfrak{A}$ is a C*-algebra satisfying the following conditions:
\begin{enumerate}
    \item $\mathfrak{A} = \text{span}\{P_1, \dots, P_{d^2}\}$ where each $P_i$ is a non-zero projection.
    \item $\sum P_i = d I$ where $I$ is the identity of $\mathfrak{A}$.
    \item For each $P,Q \in \{P_1, \dots, P_{d^2}\}$ with $P \neq Q$, we have $PQP = \frac{1}{d+1} P$.
\end{enumerate}
Then $\mathfrak{A} \cong M_d$ and $\{P_1,\dots, P_{d^2}\}$ is a SIC-POVM.
\end{theorem}

\begin{proof}
We will show that the center $Z(\mathfrak{A})$ of $\mathfrak{A}$ is the scalar multiples of the identity $I \in \mathfrak{A}$ and that $\dim(\mathfrak{A}) = d^2$. This will imply that $\mathfrak{A} \cong M_d$.

We begin by showing that $\{P_1, \dots, P_{d^2}\}$ is a linearly independent set. To this end, suppose that $\sum a_i P_i = 0$ for some scalars $a_1, \dots, a_{d^2} \in \mathbb{C}$. Conjugating $\sum a_i P_i$ by $P_j$ for some $j \in [d^2]$, we get
\[ \left( \sum_{i \neq j} \frac{a_i}{d+1} + a_j \right) P_j = 0. \]
Since $P_j \neq 0$, we see that $\sum_{i \neq j} \frac{a_i}{d+1} + a_j = 0$. Since this holds for every $j \in [d^2]$, it follows that 
\[ -a_j = \frac{1}{d+1} \sum_{i \neq j} a_i = \frac{1}{d+1} \left( \sum_{i=1}^{d^2} a_i - a_j \right) \]
and thus $a_j = \frac{1}{d} \sum_{i=1}^{d^2} a_i =: C$. So $a_j = C$ is constant. Since $0 = \sum a_i P_i = C (\sum P_i)$ and since $\sum P_i = d I$, we have $C = 0$. Therefore $\{ P_1, P_2, \dots, P_{d^2}\}$ is linearly independent.

Now suppose $T \in Z(\mathfrak{A})$. Then $[T,P_k] = 0$ for each $k$. Since $\mathfrak{A} = \text{span}\{P_i\}$, $T = \sum \alpha_i P_i$ for some scalars $\{\alpha_i\}$. For each $k \in [d^2]$,
\begin{eqnarray}
TP_k & = & P_kTP_k \nonumber \\
& = & \sum_{i=1}^{d^2} \alpha_i P_k P_i P_k \nonumber \\
& = & ( \sum_{i \neq k} \frac{\alpha_i}{d+1} + \alpha_k) P_k. \nonumber
\end{eqnarray}
Let $\lambda_k = ( \sum_{i \neq k} \frac{\alpha_i}{d+1} + \alpha_k)$ for each $k$, so that $TP_k = \lambda_k P_k$. Then
\[ T = TI = \sum_{k=1}^{d^2} \frac{1}{d} TP_k = \sum_{k=1}^{d^2} \lambda_k P_k. \]
Since $\{P_1, \dots, P_{d^2} \}$ is linearly independent, we see that $\alpha_k = \lambda_k$ for every $k \in [d^2]$. It follows that for every $k \in [d^2]$, 
\[ \left( \sum_{i=1}^{d^2} \alpha_i \right) - \alpha_k = \sum_{i \neq k} \alpha_i = \lambda_k - \alpha_k = 0. \]
So $\alpha_k = \sum_{i=1}^{d^2} \alpha_i$ for every $k \in [d^2]$. Hence
\[ T = \sum_{k=1}^{d^2} \alpha_k P_k = \left( \sum_{i=1}^{d^2} \alpha_i \right) \sum_{k=1}^{d^2} P_k = d \left( \sum_{i=1}^{d^2} \alpha_i \right) I. \]
So $T$ is a scalar multiple of $I$. It follows that $\mathfrak{A}$ is a factor. Since $\dim(\mathfrak{A}) = d^2$, $\mathfrak{A}=M_d$.

Now consider the value of $\Tr(PQ)$ for $P,Q \in \{P_1, \dots, P_{d^2}\}$. Since each $P_i$ is non-zero and since $\sum P_i = dI$, we have $\sum \Tr(P_i) = d^2$. Since $\Tr(P_i) \geq 1$ for any non-zero projection $P_i \in M_d$, we must have $\Tr(P_i) = 1$ for every $i=1,2,\dots,d^2$. So each $P_i$ is a rank one projection in $M_d$. Finally, if $P,Q \in \{P_1, \dots, P_{d^2}\}$ and $P \neq Q$, then
\[ \Tr(PQ) = \Tr(PQP) = \frac{1}{d+1} \Tr(P) = \frac{1}{d+1}. \]
We conclude that $\{P_1, \dots, P_{d^2}\}$ is a SIC-POVM. \end{proof}

We now outline how to use Theorem \ref{thm: SIC algebra} to define linear relations on a certificate $\Gamma^2$. Suppose we are given a SIC-POVM $\{P_1,\dots,P_{d^2}\}$, and consider the matrix 
\[ \Gamma_{ab,xy} = \frac{1}{d} \Tr((P_a P_b)^{\dagger} P_x P_y) \]
with $a,b,x,y \in [d^2]$. If $b = x$ and $a \neq b$, then
\begin{equation} \label{eqn: SIC relation}
\Gamma_{ab,xy} = \frac{1}{d} \Tr(P_b P_a P_b P_y) = \frac{1}{d(d+1)} \Tr(P_b P_y) = \frac{1}{d+1} \Gamma_{b,y} 
\end{equation}
for all $y \in [d^2]$. Since $M_d$ is spanned by $\{P_1, \dots, P_{d^2}\}$, we conclude that
\[ P_b P_a P_b = \frac{1}{d+1} P_b \]
whenever $a \neq b$. The next theorem says that a certificate $\Gamma^2$ satisfying Equation \ref{eqn: SIC relation} always arises from a SIC-POVM.

\begin{theorem}
Let $d \in \mathbb{N}$. Suppose that there exists a positive semidefinite matrix $(\Gamma_{v,w})_{v,w \in [d^2]^2}$ satisfying $\Gamma_{0,0}=1$, $\Gamma_{a,b} = p_{sic}^d(a,b)$ for all $a,b \in [d^2]$ and
\begin{enumerate}
    \item $\Gamma_{v,w} = \Gamma_{v',w'}$ whenever $v^{\dagger}w \sim (v')^{\dagger}w'$
    \item $\Rank(\Gamma)=d^2$
    \item $\Gamma_{ab,by} =  \frac{1}{d+1} \Gamma_{b,y}$ for all $a \neq b$ and every $y \in [d^2]$.
\end{enumerate}
Then there exists a SIC-POVM $\{P_1, \dots, P_{d^2}\} \subseteq M_d$ such that
\[ \Gamma_{ab,xy} = \frac{1}{d} \Tr((P_a P_b)^{\dagger} P_x P_y) \]
for all $a,b,x,y \in [d^2]$.
\end{theorem}

\begin{proof}
By Theorem \ref{thm: rank loop}, there exists a finite dimensional C*-algebra $\mathfrak{A}$, projections $P_1,\dots,P_{d^2} \in \mathfrak{A}$ which span $\mathfrak{A}$, and a faithful tracial state $\tau: \mathfrak{A} \to \mathbb{C}$ such that
\[ \Gamma_{ab,xy} = \tau((P_a P_b)^{\dagger} P_x P_y) \]
for all $a,b,x,y \in [d^2]$. Since $\Rank(\Gamma)=d^2$ and $\mathfrak{A}$ is spanned by $P_1, \dots, P_{d^2}$, the vectors $P_1, \dots, P_{d^2}$ must be linearly independent (hence non-zero). Since $\tau$ is faithful, $\mathfrak{A}$ may be regarded as a Hilbert space with inner product $\langle a,b \rangle := \tau(a^{\dagger}b)$ for all $a,b \in \mathfrak{A}$. Because $\mathfrak{A} = \Span \{P_1, \dots, P_{d^2}\}$, the only vector $x \in \mathfrak{A}$ satisfying $\langle x, P_i \rangle = 0$ for all $i \in [d^2]$ is $x=0$.

Now suppose $a,b \in [d^2]$ and $a \neq b$. Then for any $y \in [d^2]$,
\[ \langle P_b P_a P_b, P_d \rangle = \tau((P_a P_b)^{\dagger} P_b P_y) = \Gamma_{ab,by} = \frac{1}{d+1} \Gamma_{b,y} = \langle \frac{1}{d+1} P_b, P_y \rangle. \]
It follows that
\[ \langle P_b P_a P_b - \frac{1}{d+1} P_b, P_y \rangle = 0 \]
for all $y \in [d^2]$ and hence $P_b P_a P_b = \frac{1}{d(d+1)} P_b$. Therefore $\mathfrak{A}$ satisfies the conditions of Theorem \ref{thm: SIC algebra} and hence $\mathfrak{A} = M_d$ and $\{P_1, \dots, P_{d^2}\}$ is a SIC-POVM in $M_d$. The statement follows since $\frac{1}{d} \Tr(\cdot)$ is the unique faithful tracial state on $M_d$.
\end{proof}
\color{black}

\subsection{MUBs} \label{subsec: MUBs}

Let $H$ be a Hilbert space of dimension $d \in \mathbb{N}$. Two sets $\{ \ket{x_1} \ket{x_2}, \dots, \ket{x_d} \}$ and $\{ \ket{y_1}, \ket{y_2}, \dots, \ket{y_d} \}$ in $H$ are \textbf{mutually unbiased bases} if they are each orthonormal bases for $H$ and $|\braket{x_i}{y_j}| = \frac{1}{\sqrt{d}}$ for all $i,j \in [d]$. Letting $P_i = \ket{x_i}\bra{x_i}$ and $Q_j = \ket{y_j}\bra{y_j}$ for each $i,j \in [d]$ we obtain projection-valued measures $\{P_i\}_{i=1}^d$ and $\{Q_j\}_{j=1}^d$ which satisfy $\Tr(P_i Q_j) = \frac{1}{d}$ for all $i,j \in [d]$.

It is known that a Hilbert space of dimension $d$ can have at most $d+1$ mutually unbiased bases, or MUBs. When $d = p^n$ for some prime $p$ and some positive integer $n$, then it is also known that $d+1$ mutually unbiased bases exist. When $d$ is a composite number, it is not known if $d+1$ mutually unbiased bases exist. In particular, it is unknown whether or not there exist seven mutually unbiased bases for the Hilbert space of dimension 6, though numerical evidence suggests that no more than three MUBs exist in this Hilbert space \cite{MUBSix}.

Let $d \in \mathbb{N}$. Define 
\[ p_{mub}^d((x,i),(y,j)) = \begin{cases} \frac{1}{d} & (x,i) = (y,j) \\ 0 & x=y \text{ and } i \neq j \\ \frac{1}{d^2} & x \neq y \end{cases} \] 
for all $(x,i),(y,j) \in [d+1] \times [d]$. We now verify that $p_{mub}^d(x,y)$ extends to a positive semidefinite certificate $\Gamma^1$ satisfying $\Rank(\Gamma^1) = d^2$. The certificate $\Gamma^1$ is uniquely defined and equals 
\begin{equation}
    \Gamma^1 = 
    \begin{bmatrix} 
    1 & \vec{v}^{\dagger} & \dots & \vec{v}^{\dagger} \\
    \vec{v} & A & & B \\
    \vdots & & \ddots & \\
    \vec{v} & B & & A 
    \end{bmatrix} \in M_{d^2+d+1} \nonumber 
\end{equation}
where 
\begin{equation}
    \vec{v} = \begin{bmatrix} \frac{1}{d} \\ \vdots \\ \frac{1}{d} \end{bmatrix} \in \mathbb{M}_{d,1}, \quad
    A =
    \begin{bmatrix}
    \frac{1}{d} & & 0 \\
    & \ddots & \\
    0 & & \frac{1}{d}
    \end{bmatrix} \in \mathbb{M}_d, \text{ and }
    B = 
    \begin{bmatrix}
    \ddots & & \iddots \\
    & \frac{1}{d^2} & \\
    \iddots & & \ddots 
    \end{bmatrix} \in \mathbb{M}_d. \nonumber
\end{equation}
Here, we have written $\Gamma^1$ with respect to the enumeration 
\[ \{0, (1,1), (1,2), \dots, (1,d), (2,1), \dots, (2,d), \dots, (d+1,1), \dots, (d+1,d) \} \]
of the set of indices $\{(x,a): x \in [d+1], a \in [d]\} \cup \{0\}$ and regarding $\{P_{x,a}\}_{a=1}^d$ as a projection-valued measure for each $x \in [d+1]$. Now $\Gamma^1$ factors as
\[ \Gamma^1 = w w^T + \frac{1}{d} \begin{bmatrix} 0 & 0 & \dots & & 0 \\ 0 & A' & 0 & & \\ 0 & 0 & A' & & \\ \vdots & & & \ddots & \\ 0 & & &  & A' \end{bmatrix} \quad \text{with} \quad A' = \frac{1}{d}(I - \frac{1}{d} J) \]
where $w = \begin{bmatrix} 1 & \vec{v}^T & \dots & \vec{v}^T \end{bmatrix}^T \in \mathbb{C}^{d^2+d+1}$, $I$ denotes the $d \times d$ identity matrix and $J$ denotes the $d \times d$ matrix with every entry equal to 1. Since $I - \frac{1}{d} J$ is a rank $d - 1$ projection, $\Gamma^1$ is the sum of a rank one projection and a rank $(d+1)(d-1) = d^2-1$ projection. Hence $\Gamma^1$ is positive semidefinite. Since $A' \vec{v} = \vec{0}$, we see that $w w^T$ is orthogonal to the matrix
\[ \frac{1}{d} \begin{bmatrix} 0 & 0 & \dots & & 0 \\ 0 & A' & 0 & & \\ 0 & 0 & A' & & \\ \vdots & & & \ddots & \\ 0 & & &  & A' \end{bmatrix} \]
and hence $\Gamma^1$ is rank $d^2$.

We now wish to consider certificates $\Gamma^2$ extending $\Gamma^1$ with rank $d^2$. As in the previous subsection, we would like such a certificate to satisfy linear relations that guarantee the underlying C*-algebra generating $\Gamma^2$ is the matrix algebra $M_d$. The following theorem will allow us to find such relations.

\begin{theorem} \label{thm: MUB Algebra}
Suppose $\mathfrak{A}$ is a C*-algebra satisfying the following conditions:
\begin{enumerate}
    \item $\mathfrak{A} = \text{span}\{P_{x,a} : x \in [d+1]; a \in [d] \}$ where each $P_{x,a}$ is non-zero.
    \item $\sum_{a=1}^d P_{x,a} = I$ for each $x$, where $I$ is the identity of $\mathfrak{A}$.
    \item For each $x \neq y$, we have $P_{x,a} P_{y,b} P_{x,a} = \frac{1}{d} P_{x,a}$.
\end{enumerate}
Then $\mathfrak{A} \cong M_d$ and the projection-valued measures $\{P_{x,a}\}_{a=1}^d$ correspond to mutually unbiased bases.
\end{theorem}

\begin{proof}
We proceed as in the proof of Theorem \ref{thm: SIC algebra}, although a few details will be more tedious. We first show that $\mathfrak{A}$ has dimension $d^2$. To do this, let
\[ \mathcal{B} := \{P_{x,a} : x \in [d+1], a \in [d-1] \} \cup \{ I \}. \]
Since $\sum_{a=1}^k P_{x,a} = I$ for each $x \in [d+1]$, $\mathcal{B}$ spans $\mathfrak{A}$. We will show that $\mathcal{B}$ is a linearly independent set. To do this, suppose that
\begin{equation} \label{eqn: MUB lin ind} 
\sum_{x=1}^{d+1} \sum_{a=1}^{d-1} b_{x,a} P_{x,a} + b_0 I = 0. 
\end{equation}
For each $x \in [d+1]$ and each $a \in [d-1]$, conjugating expression \ref{eqn: MUB lin ind} by $P_{x,a}$ yields
\[ \left( \sum_{y \neq x} \sum_{c=1}^{d-1} \frac{b_{y,c}}{d} + b_{x,a}  + b_0 \right) P_{x,a} = 0. \]
Since $P_{x,a} \neq 0$, 
\begin{equation} \label{eqn: MUB 2}
\left( \sum_{y \neq x} \sum_{c=1}^{d-1} \frac{b_{y,c}}{d} \right) + b_{x,a} + b_0 = 0 
\end{equation}
for all $x \in [d+1]$ and $a \in [d-1]$. Also, conjugating expression \ref{eqn: MUB lin ind} by $P_{x,d}$ with $x \in [d+1]$ yields
\[ \left( \sum_{y \neq x} \sum_{c=1}^{d-1} \frac{b_{y,c}}{d} + b_0 \right) P_{x,a} = 0 \]
and hence
\begin{equation} \label{eqn: MUB 3}
\sum_{y \neq x} \sum_{c=1}^{d-1} \frac{b_{y,c}}{d} + b_0 = 0 
\end{equation}
for all $x \in [d+1]$. Now Equation \ref{eqn: MUB 2} together with Equation \ref{eqn: MUB 3} imply that $b_{x,a} = 0$ for every $x \in [d+1]$ and $a \in [d-1]$. This, in turn, implies that $b_0 = 0$ by Equation \ref{eqn: MUB 3}. We conclude that $\mathcal{B}$ is linearly independent. So $\dim(\mathfrak{A}) = d^2$.

We now show that $Z(\mathfrak{A})$ consists of only scalar multiplies of $I$. Suppose that $T \in Z(\mathfrak{A})$, and that
\[ T = \sum_{x=1}^{d+1} \sum_{a=1}^{d-1} \alpha_{x,a} P_{x,a} + \alpha_0 I. \]
For each $x \in [d+1]$ and $a \in [d-1]$, we have
\[ P_{x,a} T P_{x,a} = \left( \sum_{y \neq x} \sum_{c=1}^{d-1} \frac{\alpha_{y,c}}{d} + \alpha_{x,a} + \alpha_0 \right) P_{x,a} =: \lambda_{x,a} P_{x,a} \]
and, for each $x \in [d+1]$,
\[ P_{x,d} T P_{x,d} = \left( \sum_{y \neq x} \sum_{c=1}^{d-1} \frac{\alpha_{y,c}}{d} + a_0 \right) P_{x,d} =: \lambda_{x,d} P_{x,d}. \]
Now fix $x \in [d+1]$. Since $I = \sum_{a=1}^d P_{x,a}$, and since $TP_{x,a} = P_{x,a} T = P_{x,a} T P_{x,a}$ for each $a \in [d]$, we have
\[ T = TI = \sum_{a=1}^d TP_{x,a} = \sum_{a=1}^d P_{x,a} T P_{x,a} = \sum_{a=1}^d \lambda_{x,a} P_{x,a}. \]
It follows that $T \in \Span \{P_{x,1}, \dots, P_{x,d-1}, I\}$. Since this is true for every $x \in [d+1]$, and since $\mathcal{B}$ is linearly independent, we must conclude that $T = \alpha_0 I$. Therefore $\mathfrak{A}$ is a factor. Since $\dim(\mathfrak{A})=d^2$, $\mathfrak{A} \cong M_d$.

Finally, let $x,y \in [d+1]$ with $x \neq y$ and let $a,b \in [d]$. Then
\[ \Tr(P_{x,a} P_{y,b}) = \Tr(P_{x,a} P_{y,b} P_{x,a}) = \frac{1}{d} \Tr(P_{x,a}). \]
Also, since $\sum_{c=1}^d P_{x,c} = I$, we have
\[ d = \Tr(I) = \sum_{c=1}^d \Tr(P_{x,c}). \]
Since each $P_{x,c}$ is non-zero and since $\Tr(P_{x,c})$ is an integer, we conclude that $\Tr(P_{x,a}) = 1$. It follows that the set of projection-valued measures $\{P_{x,a}\}$ corresponds to family of $d+1$ mutually unbiased bases.
\end{proof}

We are now prepared to state the conditions on a certificate $\Gamma^2$ which would imply the existence of $d+1$ mutually unbiased bases in $\mathbb{C}^d$. To do so, we will need to describe a matrix $\Gamma$ indexed by words in the letters $\{(x,a) : x \in [d+1], a \in [d]\}$. To simplify notation, let $A_{x,d}$ denote the set of symbols $\{(x,1),\dots,(x,d)\}$ and let $A_d = \cup_{x=1}^{d+1} A_{x,d}$.

\begin{theorem}
Let $d \in \mathbb{N}$. Suppose that there exists a positive semidefinite matrix $\Gamma^2 = (\Gamma_{v,w})$ indexed by words in $A_d^2$ satisfying $\Gamma_{0,0}=1$, $\Gamma_{a,b} = p_{mub}^d(a,b)$ for all $a,b \in A_d$ and
\begin{enumerate}
    \item $\Gamma_{v,w} = \Gamma_{v',w'}$ whenever $v^{\dagger}w \sim (v')^{\dagger}w'$
    \item $\Rank(\Gamma^2)=d^2$ 
    \item whenever $a \in A_{x,d}$, $b \in A_{y,d}$ with $x \neq y$, and $c \in A_d$, we have $\Gamma_{ab,bc} =  \frac{1}{d} \Gamma_{b,c}$.
\end{enumerate}
Then there exists, for each $x \in [d+1]$, a projection-valued measure $\{P_a \}_{a \in A_{x,d}} \subseteq M_d$, and 
\[ \Gamma_{ab,a'b'} = \frac{1}{d} \Tr((P_a P_b)^{\dagger} P_{a'} P_{b'}) \]
for all $a,b,a',b' \in A_d$. In particular, there exist $d+1$ mutually unbiased bases in $\mathbb{C}^d$.
\end{theorem}

\begin{proof}
By Theorem \ref{thm: rank loop}, there exists a finite dimensional C*-algebra $\mathfrak{A}$, projections $\{ P_{a} : a \in A_d \} \in \mathfrak{A}$ which span $\mathfrak{A}$, and a faithful tracial state $\tau: \mathfrak{A} \to \mathbb{C}$ such that
\[ \Gamma_{ab,xy} = \tau((P_a P_b)^{\dagger} P_x P_y) \]
for all $a,b,x,y \in A_d$. Since $\Rank(\Gamma^2)=d^2$ and $\Gamma_{0,a} = \tau(P_a) = \frac{1}{d}$ for each $a \in A_d$, each vector $P_a$ must be non-zero. Since $\tau$ is faithful, $\mathfrak{A}$ may be regarded as a Hilbert space with inner product $\langle a,b \rangle := \tau(a^{\dagger}b)$ for all $a,b \in \mathfrak{A}$. Because $\mathfrak{A} = \Span \{ P_a : a \in A_d \}$, the only vector $x \in \mathfrak{A}$ satisfying $\langle x, P_a \rangle = 0$ for all $a \in A_d$ is $x=0$.

Now suppose $a \in A_x$ and $b \in A_y$ and $x \neq y$. Then for any $c \in A_d$,
\[ \langle P_b P_a P_b, P_c \rangle = \tau((P_a P_b)^{\dagger} P_b P_c) = \Gamma_{ab,bc} = \frac{1}{d} \Gamma_{b,c} = \langle \frac{1}{d} P_b, P_c \rangle. \]
It follows that
\[ \langle P_b P_a P_b - \frac{1}{d} P_b, P_y \rangle = 0 \]
for all $y \in [d^2]$ and hence $P_b P_a P_b = \frac{1}{d} P_b$. Therefore $\mathfrak{A}$ satisfies the conditions of Theorem \ref{thm: MUB Algebra} and hence $\mathfrak{A} \cong M_d$, $\tau = \frac{1}{d} \Tr$, and for every $x \in [d+1]$, $\{P_a : A_x\}$ is a projection valued measure consisting of rank one projections. It follows that the families $\{P_a : A_x\}$ for $x \in [d+1]$ correspond to $d+1$ mutually unbiased bases in $M_d$.
\end{proof}

\bibliographystyle{plain}
\bibliography{references}

\end{document}